\documentclass[10pt,conference]{IEEEtran}
\usepackage{amsfonts}
\usepackage{amsmath}
\usepackage{amssymb}
\usepackage{amsthm}
\usepackage{lscape}
\usepackage{epsf}
\usepackage{graphicx}
\usepackage{verbatim}
\usepackage{color} 
\usepackage{algorithm}
\usepackage{algpseudocode}
\usepackage[top=54pt,right=54pt,left=54pt,bottom=54pt,letterpaper]{geometry}

\newcommand{\xx}{\mathbf{x}}
\DeclareMathOperator{\rank}{rank}
\DeclareMathOperator*{\rowspace}{rowspace}

\newtheorem{theorem}{\indent Theorem}[section]
\newtheorem{lemma}[theorem]{\indent Lemma}
\newtheorem{corollary}[theorem]{\indent Corollary}

\newtheorem{proposition}[theorem]{\indent Proposition}

\newtheorem{example}{\indent Example}[section]
\newtheorem{definition}{\indent Definition}[section]

\newtheorem{remark}{\indent Remark}[section]

\newcommand{\FFF}{\mathbb{F}}
\newcommand{\ffs}{{\mathbb{F}^\star}}

\newcommand{\AAA}{\mathbb{A}}

\DeclareMathOperator*{\diag}{diag}

\DeclareMathOperator*{\MAXRANK}{max-rank}
\DeclareMathOperator{\MINRANK}{min-rank}

\newcommand{\graph}{{\mathcal{G}}}

\newcommand{\cV}{{\mathcal{V}}}
\newcommand{\cH}{{\mathcal{H}}}
 
\newcommand{\cE}{{\mathcal{E}}}

\newcommand{\cB}{{\mathcal{B}}}
\newcommand{\cA}{{\mathcal{A}}}

\newcommand{\cN}{{\mathcal{N}}}

\newcommand{\cP}{{\mathcal{P}}}
\newcommand{\cQ}{{\mathcal{Q}}}

\newcommand{\cT}{{\mathcal{T}}}

\newcommand{\distance}{{\mathsf{d}}}

\newcommand{\sP}{{\mathsf{P}}}

\newcommand{\clc}{{\mathsf{cc}}}

\newcommand{\ff}{{\mathbb{F}}}

\newcommand{\nn}{{\mathbb N}}
\newcommand{\pp}{{\mathbb P}}

\newcommand{\bldA}{{\mbox{\boldmath $A$}}}
\newcommand{\bldAA}{{\mbox{\scriptsize \boldmath $A$}}}

\newcommand{\bldB}{{\mbox{\boldmath $B$}}}
\newcommand{\bldb}{{\mbox{\boldmath $b$}}}

\newcommand{\bldD}{{\mbox{\boldmath $D$}}}
\newcommand{\bldDD}{{\mbox{\scriptsize \boldmath $D$}}}

\newcommand{\blde}{{\mbox{\boldmath $e$}}}
\newcommand{\bldE}{{\mbox{\boldmath $E$}}}
\newcommand{\bldEE}{{\mbox{\scriptsize \boldmath $E$}}}

\newcommand{\bldI}{{\mbox{\boldmath $I$}}}

\newcommand{\blds}{{\mbox{\boldmath $s$}}}

\newcommand{\bldM}{{\mbox{\boldmath $M$}}}
\newcommand{\bldT}{{\mbox{\boldmath $T$}}}
\newcommand{\bldUpsilon}{{\mbox{\boldmath $\Upsilon$}}}

\newcommand{\bldP}{{\mbox{\boldmath $P$}}}

\newcommand{\bldv}{{\mbox{\boldmath $v$}}}
\newcommand{\bldu}{{\mbox{\boldmath $u$}}}
\newcommand{\bldw}{{\mbox{\boldmath $w$}}}
\newcommand{\bldx}{{\mbox{\boldmath $x$}}}

\newcommand{\bldy}{{\mbox{\boldmath $y$}}}

\newcommand{\tauopt}{{\mbox{$\tau_{\mbox{\scriptsize opt}}$}}}

\newcommand{\bldzero}{{\mbox{\boldmath $0$}}}

\newcommand{\bldone}{{\mathbf{1}}}

\title{Data Dissemination Problem in Wireless Networks}
\author{{\bf Ivo Kubjas and Vitaly Skachek} \\ 
Institute of Computer Science \\
University of Tartu, Tartu 50409, Estonia
}

\begin{document}

\addtolength{\topmargin}{18pt}
\addtolength{\textheight}{-18pt}
\maketitle

\begin{abstract}
In this work, we formulate and study a data dissemination problem, which can be viewed as a generalization of 
the index coding problem and of the data exchange problem to networks with an arbitrary topology. 
We define $r$-solvable networks, in which data dissemination can be achieved in $r > 0$ communications rounds. 
We show that the optimum number of transmissions for any one-round communications scheme is given 
by the minimum rank of a certain constrained family of matrices. 
For a special case of this problem, called bipartite data dissemination problem, we present lower and upper 
graph-theoretic bounds on the optimum number of transmissions. 
For general $r$-solvable networks, 
we derive an upper bound on the minimum number of transmissions in any scheme with $\geq r$ rounds. 
We experimentally compare the obtained upper bound to a simple lower bound.
\end{abstract}


\section{Introduction}

A problem of \emph{index coding with side information} considers a communications scenario with one broadcast transmitter 
and several receivers. All receivers possess some partial information available to the transmitter and request additional information. 
The goal is to design a communications scheme, which minimizes the total number of transmissions. 

Index coding problem was proposed first in~\cite{BirkKol2006}: it was suggested therein to use coding in order to minimize a 
number of transmissions. Later, in~\cite{Yossef-journal}, the minimum number of transmissions in the index coding problem 
was shown to be equal to the minimum rank of a properly defined family of matrices. Generally, computing the minimum rank of a family of the 
matrices is an NP-hard problem, yet in some special cases there exist efficient algorithms to compute it~\cite{Yossef-journal,Dau-2014}.

Index coding problem was intensively studied in the recent years, see for example~\cite{Alon, ChaudhrySprintson,Dau-2012,Dau-2013,Katti2006,Maddah-Ali-2013}. It was shown in~\cite{Effros, Rouayheb2009} that index coding problem is equivalent to a network coding 
problem~\cite{Ahlswede, KoetterMedard2003}. In index coding, however, the underlying network graph is very simple, it is a directed star graph,
where the transmitter is the root of that graph.  

A variation of index coding, termed \emph{data exchange problem}, was studied in~\cite{ELRouhayeb-2010}. In the data exchange problem, unlike the index coding problem, every node can serve as both a transmitter and a receiver. The underlying network graph is a complete directed graph. 
Before the communications take place, each node possesses some partial information. The goal is to deliver all information to all 
the nodes in a minimum number of transmissions. It was shown in~\cite{ELRouhayeb-2010} that the minimum number of transmissions in 
the data exchange problem can also be described as a rank minimization problem of a certain constrained family of matrices, thus 
resembling some of the results for index coding.    

Another related problem is a \emph{set reconciliation}~\cite{Eppstein, Gabrys, Ivo-thesis, Minsky, Mitzenmacher, Skachek-Rabbat}. 
The set reconciliation problem is usually defined over a network of arbitrary topology, either wired or wireless. 
In that problem, similarly to the data exchange problem, the goal is to deliver all information to all the nodes. 
However, by contrast, it is assumed that no node knows what information is possessed by the other nodes. This makes the set reconciliation problem 
more difficult than the data exchange problem. 

In this work, we introduce a \emph{data dissemination problem}, which further generalizes both the index coding and the data exchange problems, such that the underlying directed connectivity graph of the network is an arbitrary graph. This model, in particular, represents 
cached networks of arbitrary topology. The data dissemination problem can also be viewed as a generalization of the set reconciliation 
problem. In data dissemination problem, every node can serve as both a transmitter and a receiver. Moreover, each node possesses some partial information and requests some additional information. 

\begin{example}
Consider an example network in Figure~\ref{ex:network}. 
There are five nodes $v_1$, $v_2$, $v_3$, $v_4$ and $v_5$, which in total possess three bits of information $x_1$, $x_2$, $x_3$. If $v_1$ transmits $x_1 + x_2$
and $v_2$ transmits $x_2 + x_3$, then the requests of all nodes will be satisfied with only two transmissions.

\begin{figure}[htb]
\begin{picture}(100, 140)(20,-15)

\put(140, 20){\circle*{4}}
\put(60, 20 ){\circle*{4}}
\put(220, 20){\circle*{4}}

\put(100, 100){\circle*{4}}
\put(180, 100){\circle*{4}}

\put(100, 100){\vector(1,-2){40}}
\put(180, 100){\vector(1,-2){40}}
\put(100, 100){\vector(-1,-2){40}}
\put(180, 100){\vector(-1,-2){40}}
\put(180, 100){\vector(-3,-2){120}}

\small

\put(96, 114){\makebox{$v_1$}}
\put(80, 105){\makebox{has $x_1$, $x_2$}}

\put(176, 114){\makebox{$v_2$}}
\put(160, 105){\makebox{has $x_2$, $x_3$}}

\put(58, 10){\makebox{$v_3$}}
\put(50, 2){\makebox{has $x_1$}}
\put(32, -6){\makebox{requests $x_2$, $x_3$}}

\put(138, 10){\makebox{$v_4$}}
\put(130, 2){\makebox{has $x_2$}}
\put(112, -6){\makebox{requests $x_1$, $x_3$}}

\put(218, 10){\makebox{$v_5$}}
\put(202, 2){\makebox{has $x_1$, $x_3$}}
\put(202, -6){\makebox{requests $x_2$}}

\put(100, 87){\oval(50, 30)[b]}
\put(20, 90){\makebox{transmits $x_1 + x_2$}}

\put(180, 87){\oval(50, 30)[b]}
\put(190, 90){\makebox{transmits $x_2 + x_3$}}
\end{picture}
\caption{Example network}
\label{ex:network}
\end{figure}
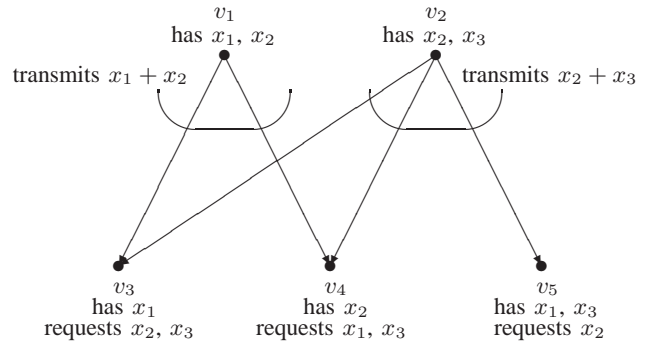
\end{example}

In this work, we present the following results.
First, we formulate and study a data dissemination problem. 
We define $r$-solvable networks, in which data dissemination can be achieved in $r > 0$ communications rounds. 
We show that the optimal number of transmissions for any one-round communications scheme is given 
by the minimum rank of a certain constrained family of matrices. 
For a special case of this problem, termed bipartite data dissemination problem, we present lower and upper 
graph-theoretic bounds on the optimum number of transmissions. 
For general $r$-solvable networks, 
by using similar techniques, we derive an upper bound on the minimum number of transmissions in $\ge r$ rounds. 
We experimentally compare the obtained upper bound to a simple lower bound.

\section{Notation}

\addtolength{\topmargin}{-18pt}
\addtolength{\textheight}{18pt}

Denote $[n] \triangleq \{1, 2, \cdots, n\}$ (in particular, $[0]$ denotes the empty set). We use $\bldzero$ to denote the all-zero vector, when the length of the vector is clear from the context. Similarly, we use $\blde_i$ to denote the unit vector which has $1$ in position $i$ and zeros everywhere else. We assume hereafter that all vectors are column vectors. 

Let $\ff$ be a finite field $\ff_q$, where $q$ is a prime power.  
Take $\bldA$ to be a matrix over $\ff$. Denote by $\bldA^{[i]}$ the $i$-th row of $\bldA$ and 
by $(\bldA)_{i,j}$ the entry in the $i$-th row and $j$-th column of $\bldA$. 
We use the notation $\rowspace(\bldA)$ to denote the row space of the matrix $\bldA$, and notation $\bldA \otimes \bldB$ for the 
standard tensor product of the matrices $\bldA$ and $\bldB$. For the row vector $\bldv = (v_1, v_2, \cdots, v_n)$, we denote 
by $\diag(\bldv)$ the $n \times n$ matrix as follows: 
\[
\left( \diag(\bldv) \right)_{i,j} = \left\{ \begin{array}{cc}
v_i & \mbox{if } i = j \\
0 & \mbox{otherwise}
\end{array} \right. \; . 
\]

Fix an ambient vector space $V \subseteq \ff^n$. Let $W$ be a subspace of $V$. 
The orthogonal vector space of $W$ is given by 
    \[
        W^\perp \triangleq \{ \bldv \in V \; | \; \forall \bldw \in W \, : \, \bldv \cdot \bldw = \bldzero \} \; , 
    \]
where $\bldv \cdot \bldw$ denotes the inner product of the two vectors. 

Let $U, W \subseteq V$ be two vector subspaces. Define
\[
U + W = \{ \bldu + \bldw \; | \; \bldu \in U \mbox{ and } \bldw \in W \} \subseteq V \; . 
\]
If $U \cap W = \{ \bldzero \}$, then we also write $U \oplus W$ instead of $U + W$. 

Let $\graph(\cV, \cE)$ be a directed graph with the vertex set $\cV$ and the edge set $\cE$. For each $\ell \in \cV$, introduce the notations
\begin{multline*}
\cN_{in} (\ell) = \{ v \in \cV \; : \; (v, \ell) \in \cE \} \\
\mbox{ and } \quad 
\cN_{out} (\ell) = \{ v \in \cV \; : \; (\ell, v) \in \cE \} \; . 
\end{multline*}

Let $\bldE$ be the all-one square matrix. The size of $\bldE$ will be apparent from the context. 
Similarly, let $\bldI$ be the identity matrix. Finally, denote by $\bldone_n$ the all-one column vector of length $n$.

\section{Problem setup}

First, we define the data dissemination problem.
Consider a wireless network with a topology given by a finite directed connected graph $\graph(\cV, \cE)$, 
where $\cV = [k]$ is the set of nodes and $\cE$ is the set of edges. 
Let $\xx=(x_1, \ldots x_n) \in \ff^n$ be an information vector. 
Each node $\ell \in \cV$ possesses some side information consisting of the symbols $x_j$,   
$j \in \cP_\ell \subseteq [n]$, and is interested in receiving all of the symbols $x_i$, 
$i \in \cT_\ell \subseteq [n] \backslash \cP_\ell$. 

In the data dissemination problem, the goal is to find a coded transmission
protocol with the minimum number 
of transmissions, such that all nodes could recover all their respective requested symbols. 
However, unlike in~\cite{ELRouhayeb-2010}, the network might not have full connectivity. 

Throughout this paper we make the following assumptions. 
\begin{itemize}
\item
The graph $\graph$ is an arbitrary directed graph. 
\item
All transmissions are broadcast, i.e. the messages transmitted by the node $\ell$ are always received by all the nodes in $\cN_{out}(\ell)$.
\item
The coding is linear, i.e. each node $\ell \in \cV$ transmits messages of the form $\sum_{i \in [n]} \mu_i \cdot x_i$, where $\mu_i \in \ff$ for all $i \in [n]$. 
\item
There is a central entity that knows $\graph$ and all the sets $\cP_i$ and $\cT_i$ for all $i \in \cV$. This entity is running an algorithm for finding an optimal communications scheme. 
\item The transmissions are without errors and interference, i.e. all transmissions are received correctly. This can be achieved by separately handling error-correction and interference in the lower layers. 
\item There are no parallel edges in the network graph. This assumption simplifies the analysis, yet similar analysis can be done if the graph has parallel edges.  
\end{itemize}

Hereafter, we also assume that the transmissions are \emph{performed in rounds}. 
During a round, each node transmits linear combinations of the symbols it possesses, 
but it can not use any symbol received in the same round. 

In Figure~\ref{fig:protocol}, we formally define Protocol $\sP$. Let $r$ be an integer parameter, which denotes the number of rounds in the protocol. For all $\ell \in \cV$ and for all $i = 0, 1, 2, \cdots, r$, denote by $\cQ_\ell^{(i)}$ the 
set of symbols in $\ff$ possessed by the node $\ell$ at the end of round $i$ (here, $\cQ_\ell^{(0)} \triangleq \{ x_j \}_{j \in \cP_\ell}$ denotes the set of symbols 
possessed by the node $\ell$ before the execution of the protocol). 

\begin{figure}[hbt]
\makebox[0in]{}\hrulefill\makebox[0in]{}
{\normalsize\sl
\begin{description}
\item[{\bf For $i = 1$ to $r$}] \hspace{8ex} {\bf do} \{  
\begin{description}
\item[\underline{\bf Transmitting phase:}] \hspace{13ex} 
each node $\ell \in \cV$ broadcasts $\tau_{i,\ell}$ linear combinations 
$z_{i,\ell,j} = \sum_{x \in \cQ_\ell^{(i-1)}} \mu_{x, i, \ell, j} \cdot x$, \; $j = 1, 2, \cdots, \tau_{i,\ell}$.
\item[\underline{\bf Receiving phase:}] \hspace{9.7ex}
each node $\ell \in \cV$ updates $\cQ_\ell^{(i)} = \cQ_\ell^{(i-1)} \cup 
\{ z_{i,v,j} \}_{ v \in \cN_{in}(\ell), \; j = 1, 2, \cdots, \tau_{i,v}}$.  
\end{description}
\item[] \}
\item[\underline{\bf Recovery phase:}] \hspace{9ex} 
each node $\ell \in \cV$ computes 
$x_j = \sum_{x \in \cQ_\ell^{(r)}} \mu_{x, \ell, j} \cdot x$ for all $j \in \cT_\ell$.
\end{description}
}
\makebox[0in]{}\hrulefill\makebox[0in]{}
\caption{Protocol $\sP$.}
\label{fig:protocol}
\end{figure}

Here, the integers $\tau_{i,\ell}$ denote the number of symbols in $\ff$ transmitted by
the node $\ell \in \cV$ in round $i \in [r]$. The coefficients $\mu_{x, i, \ell, j} \in \ff$ are 
chosen to multiply the symbols $x \in \cQ_\ell^{(i-1)}$ in the $j$-th linear combination, $j \in [\tau_{i,\ell}]$, $\ell \in \cV$, $i \in [r]$,
and similarly $\mu_{x, \ell, j} \in \ff$ are chosen to multiply the symbols $x \in \cQ_\ell^{(r)}$ in the $j$-th linear combination, $j \in \cT_\ell$.
\medskip 

We note that both the index coding problem~\cite{BirkKol2006} and the data exchange problem~\cite{ELRouhayeb-2010} can be 
viewed as special cases of Protocol $\sP$, for $r = 1$. 

\begin{definition}
\label{def:feasible}
    Consider a network based on the graph $\graph(\cV, \cE)$. 
		The assignment of the sets $\cP_i$ and $\cT_i$ is called \emph{feasible} 
		if for any $j \in \cT_i$, $i \in \cV$, there exists 
		a node $\ell \in \cV$ with $j \in \cP_\ell$, such that there is a finite directed path from $\ell$ to $i$ in $\graph$. 
\end{definition}
It readily follows from Definition~\ref{def:feasible} that for any feasible assignment there exists a selection of the integers $\tau_{i,\ell}$ and of the coefficients $\mu_{x, i, \ell, j}$ and $\mu_{x, \ell, j}$ in Protocol $\sP$, such that all requests are satisfied in a finite number of rounds (for example, by simply forwarding the requested bits in $\graph$, 
so all the coefficients are either 0 or 1). 
		
\begin{definition}
\label{def:r-solvable}
    The network based on the graph $\graph(\cV, \cE)$ is said to be $r$-\emph{solvable}, $r \in \nn$,
		if it is strongly connected and for any feasible assignment of the sets $\cP_i$ and $\cT_i$, $i \in \cV$, $r$ communications rounds are sufficient for the protocol to satisfy all the node requests, but $r-1$ rounds are not sufficient. If the network is not $r$-solvable for any $r \in \nn$, then we say that it is \emph{not solvable}. 
\end{definition}

\begin{lemma} 
The network is $r$-solvable for some $r \in \nn$ if the maximum of the shortest length of the directed path from the node $i$ to the node $\ell$ in $\graph$, for any two nodes $i, \ell \in \cV$, is exactly $r$. 
\label{lemma:shortest-path}
\end{lemma}
\vspace{-1ex}
The proof of this lemma appears in the appendix. 
\medskip

We define the transposed $k \times k$ integer adjacency matrix $\bldD$ of the graph $\graph(\cV,\cE)$ as follows: 
\[
    (\bldD)_{i,j} = \left\{ \begin{array}{cc}
		1 & \mbox{if } (j, i) \in \cE \\
		0 & \mbox{otherwise} 
		\end{array} \right. \; . 
\]
\begin{corollary}
The network is $r$-solvable if $r$ is the smallest integer such that all the entries in the matrix $\bldD^r$ are strictly positive.  
\end{corollary}

\section{Optimal solution for one round data exchange problem}
\label{sec:one_round}

Let the graph $\graph(\cV, \cE)$, the information vector $\xx=(x_1, \ldots x_n) \in \ff^n$, and 
the sets $\cP_\ell$ and $\cT_\ell$ for $\ell \in \cV$ be defined as above. 
We represent a matrix family $\AAA$ over $\ff$ as a matrix over $\ff \cup \{ \mbox{`$\star$'} \}$, where
`$\star$' is a special symbol. The entry, which can take any value from $\ff$ in $\AAA$ is 
marked as `$\star$'. 

For each node $\ell \in \cV$, we define the family $\AAA_\ell$ of $n \times n$ matrices as follows. 
\begin{equation}
\left( \AAA_\ell \right)_{i, j} = \left\{ \begin{array}{cc}
\mbox{`$\star$'} & \mbox{if } j \in \cP_\ell \\
0 & \mbox{otherwise}
\end{array} \right. \; . 
\label{eq:A-i}
\end{equation}
Define the family $\AAA$ of $(k n) \times n$ block matrices as:
\begin{equation}
    \AAA \triangleq  \left[ \begin{array}{c} \AAA_1 \\ \hline \AAA_2 \\ \hline \vdots \\ \hline \AAA_{k} \end{array} \right] \; .
		\label{eq:possession-matrix}
\end{equation}
Given $\bldA \in \AAA$, the $j$-th $n \times n$ sub-matrix of $\bldA$ will be denoted as $\bldA_j$. We will also use the notation 
$\bldA_{\cN_{in}(\ell)}$ to denote the $d n \times n$ matrix 
\begin{equation*}
\bldA_{\cN_{in}(\ell)} = \left[ \begin{array}{c}
\bldA_{i_1} \\
\hline
\bldA_{i_2} \\
\hline
\vdots \\
\hline
\bldA_{i_d} 
\end{array} \right] \; ,
\end{equation*}
where $\cN_{in}(\ell) = \{ i_1, i_2, \cdots, i_d \}$, and $d$ is an in-degree of $\ell$ in $\graph$.  

For each $\ell \in \cV$, we define an $n \times n$ \emph{information matrix} $\bldP_\ell = (\bldP_\ell)_{i \in [n], j \in [n]}$, 
\[
(\bldP_\ell)_{i, j} = \left\{ \begin{array}{ll}
1 & \mbox{if } i = j \mbox { and } i \in \cP_\ell \\
0 & \mbox{otherwise}
\end{array} \right. \; . 
\]

Similarly, for each $\ell \in \cV$, we define an $n \times n$ \emph{query matrix} $\bldT_\ell = (\bldT_\ell)_{i \in [n], j \in [n]}$, 
\[
(\bldT_\ell)_{i, j} = \left\{ \begin{array}{ll}
1 & \mbox{if } i = j \mbox { and } i \in \cT_\ell \\
0 & \mbox{otherwise}
\end{array} \right. \; . 
\]

In \cite{Yossef-journal}, it was shown that if the network is a star graph, then
the min-rank of the side information graph yields the optimal number of
transmissions. We extend this result to
all networks where the transmitters and receivers may possess and request
arbitrary information bits such that it is possible to perform the transmissions
in one round.

\begin{theorem}
\label{thm:coding}
    Consider a wireless network defined by the graph $\graph(\cV, \cE)$.
		Let $\AAA$ be an $nk \times n $ matrix family defined as above. For all nodes $\ell \in \cV$,
    let $\bldP_\ell$ and $\bldT_\ell$ be the corresponding possession and query matrices. Then, the
    minimal number of transmissions in Protocol $\sP$ needed to satisfy the demands of all nodes in 
		$\cV$ in \emph{one round} of communications is 
    \begin{equation}
        \tau = \min_{\bldAA \in \AAA} \left\{ \sum_{\ell \in \cV} \rank \left( \bldA_\ell \right) \right\} \; ,
				\label{eq:tau-expression}
    \end{equation}
    where for all $\ell \in \cV$
    \begin{equation}
        \rowspace \left( \left[ \begin{array}{c} \bldA_{\cN_{in}(\ell)} \\ \hline \bldP_\ell \end{array} \right]
        \right) \supseteq\rowspace(\bldT_\ell) \; . 
				\label{ex_condition}
    \end{equation}
If the above matrix $\bldA \in \AAA$ as above does not exist then there is no algorithm that satisfies all requests in one round. 
\end{theorem}
\vspace{-1ex}
The proof of this theorem appears in the appendix. 
\medskip

It is straightforward to see that when $\graph$ is a complete directed graph, then the result of Theorem~\ref{thm:coding} is equivalent 
to what is shown for the data exchange problem in~\cite[Section V]{ELRouhayeb-2010}. When $\graph$ is a directed star graph, then 
this result is equivalent to the min-rank optimization for the index coding problem as in~\cite{Yossef-journal}.

\section{Graph-theoretic bounds}

% Definition of side information graph

In this section, we consider a restricted special case of the data dissemination problem, termed \emph{bipartite data dissemination problem}, where 
the underlying network graph $\graph(\cV, \cE)$ is a bipartite graph, with $\cV = \cA \cup \cB$, $\cA \cap \cB = \varnothing$, 
where the set $\cA$ is a set of transmitters, the set $\cB$ is the set of receivers, and $\cE \subseteq \cA \times \cB$. 
Additionally, assume that the information vector is $\xx=(x_1, \ldots x_n) \in \ff^n$, and $|\cB| = n$.   
For all $\ell \in \cA$ we have $\cP_\ell = [n]$ and $\cT_\ell = \varnothing$, i.e. the nodes in $\cA$ 
serve as transmitters only, and they all have all $x_i$'s. 
We assume w.l.o.g. that $\cB = [n]$, and furthermore, for the sake of simplicity, we also assume that 
for each $\ell \in \cB$, $\cT_\ell = \{ \ell \}$. 
Observe that the case, when $|\cA| = 1$ and there exists an edge $(i,\ell) \in \cE$ for $i \in \cA$ and all $\ell \in \cB$, 
corresponds to the well-known \emph{index coding} problem~\cite{BirkKol2006}. 

Next, we recall some known facts from the literature. 

\begin{definition}[\cite{Yossef-journal}]
    \label{def:side_information_graph}
    For an index coding instance as above, a \emph{side information graph} is the directed graph $\cH = (\cV_\cH, \cE_\cH)$ with the vertex set $\cV_\cH = [n]$
    and the edge set $\cE_\cH= \left\{ (i,j) \; | \; j \in \cP_i \right\}$.
\end{definition}

If the side information graph is symmetric, i.e.
\[
\forall i, j \in [n] : \, (i, j) \in \cE_\cH \quad \Leftrightarrow \quad (j, i) \in \cE_\cH \;  , 
\]
then we can view it as the corresponding undirected graph. 

\begin{definition}[\cite{Haemers}]
    \label{def:graph_fitting}
    Let $\cH$ be a graph with a vertex set $\cV_\cH = [n]$. We say that a $n \times
    n$-dimensional $0-1$ matrix $\bldM$ fits $\cH$ if for all $i \in [n]$ we have $(\bldM)_{i,i}=1$, and
    for all $i,j \in [n]$, $(\bldM)_{i,j} = 0$ if there is no edge $(i, j)$ in $\cH$.

    The minimum rank of graph $\graph$ is defined as 
    \[
        \MINRANK_2\left(\cH\right) \triangleq \min_{\bldM} \left\{\rank_2(\bldM) | \mbox{ matrix }
        \bldM \mbox{ fits } \cH \right\} \; ,
    \]
		where $\rank_2(\bldM)$ denotes the rank of the binary matrix $\bldM$ over $\ff_2$. 
\end{definition}

% Reference to bounds using SIG

\begin{theorem}[\cite{Yossef-journal}]
    \label{thm:minrank_bounds}
    For any undirected graph $\cH$,
    \[
        \alpha \left(\cH \right) \leq \Theta\left(\cH\right) \leq
        \MINRANK_2\left(\cH\right) \leq \clc \left( \cH \right) \; , 
    \]
		where $\alpha\left(\cH \right)$ is the size of the largest independent set in $\cH$, 
		$\Theta\left(\cH\right)$ is the Shannon capacity of $\cH$~\cite{Haemers}, and $\clc \left( \cH \right)$ is the minimum click cover size of the graph $\cH$. 
\end{theorem}

\begin{theorem}[\cite{Yossef-journal}]
    \label{thm:optimal_index}
		Consider an instance of index coding problem over $\ff_2$ represented by a side information graph $\cH$.
    The minimum number of transmissions in any solution for this instance using linear code is given 
		by $\MINRANK_2\left( \cH \right)$.
\end{theorem}

Take an instance of a bipartite data dissemination problem. We define a side information graph $\cH$ 
for that problem analogous to Definition~\ref{def:side_information_graph}, namely
a side information graph $\cH = (\cV_\cH, \cE_\cH)$ is a graph with the vertex set $\cV_\cH = [n]$
    and the edge set $\cE_\cH = \left\{ (i,j) \; | \; j \in \cP_i \right\}$.

Below, we derive new graph-theoretic 
upper and lower bounds on the optimal number of transmissions for the bipartite data dissemination problem. We start with the lower bound, given by the following lemma. 

\begin{proposition}
The optimal number of transmissions in Protocol $\sP$ for a bipartite data dissemination problem is at least $\MINRANK_2\left( \cH \right)$, 
which is in turn bounded from below by $\Theta\left(\cH\right)$ and by $\alpha \left( \cH \right)$.
\end{proposition}

\begin{proof}
We convert the given instance of a bipartite data dissemination problem into an instance of index coding, by replacing all the transmitters in $\cA$ by one super-transmitter, and by adding edges from this super-transmitter to all the nodes in $\cB$. The optimal transmission scheme for the original data dissemination problem is also a valid transmission scheme for the corresponding index coding problem. The side information graph for both instances is exactly the same. The optimal number of transmissions in the later scheme 
is bounded from below by $\MINRANK_2\left( \cH \right)$ due to Theorem~\ref{thm:optimal_index}, and the claim follows. 
\end{proof}
\medskip 

Next, consider an instance of a bipartite data dissemination problem with 
the underlying network graph $\graph(\cV, \cE)$, where $\cV = \cA \cup \cB$, $\cA \cap \cB = \varnothing$, 
and $\cE \subseteq \cA \times \cB$.
Assume that $|\cA| = t$ and let $\mathfrak{P} = \left\{ \graph_1, \graph_2, \cdots , \graph_t \right\}$, $\graph_i = (\cV_i, \cE_i)$, 
be a set of graphs induced by the partition $\cV = \cup_{i \in [t]} \cV_i$, all $\cV_i$ are disjoint, such
that for all $i \in [t]$ there is $|\cA \cap \cV_i| = 1$. Since there are no parallel edges, we obtain that $|\left\{ (v,b) |
(v,b) \in \cE_i \right\}|=1$ for all $i \in [t]$ and all $b \in \cB \cap \cV_i$.

\begin{lemma}
For all $\graph_i$ as above, $i \in [t]$, consider an induced instance of index coding problem with a transmitter in $\cA \cap \cV_i$
and the set of the receivers $\cB \cap \cV_i$. 
For each $\ell \in \cV_i$, the sets $\cP_\ell$ and $\cT_\ell$ are defined exactly as in the original problem.
Denote by $\cH_i$, $i \in [t]$, the corresponding side information graph. 
Then, the optimum number of transmissions in Protocol $\sP$ for the given bipartite data dissemination problem is less or equal to
\[
\sum_{i=1}^t \MINRANK_2 \left( \cH_i \right) \; .
\]
\end{lemma}
\begin{proof}
Let $\graph(\cV, \cE)$ be an underlying network graph of a bipartite data dissemination problem, and let $\mathfrak{P}$
be as defined above. We construct a new bipartite data dissemination problem, as follows. The network graph $\tilde{\graph}(\cV, \tilde{\cE})$
for the new problem is given by 
\[
\tilde{\cE} = \cup_{i \in [t]} \cE_i \subseteq \cE\; . 
\]
For each $\ell \in \cV$, the sets $\cP_\ell$ and $\cT_\ell$ are defined exactly as in the original problem. 

The optimum number of transmissions for the original problem is less or equal to the number of transmissions in any solution for the new problem, 
because the optimum solution for the new problem is a valid solution to the original problem. Since there are no edges connecting nodes in different graphs $\graph_i$, the optimum solution to the new problem is obtained as a combination of optimum solutions to each of the 
index coding problems induced by the graphs $\graph_i$, $i \in [t]$. Therefore, the optimum number of solutions to the new problem is given by 
\[
\sum_{i=1}^t \MINRANK_2 \left( \cH_i \right) \; ,
\]
and it serves as an upper bound on the optimum number of transmissions for the original problem. 
\end{proof}
\medskip 

Denote by $\pp$ the set of all graph partitions $\mathfrak{P} = \left\{ \graph_1, \graph_2, \cdots , \graph_t \right\}$, $\graph_i = (\cV_i, \cE_i)$, 
such that $\cV = \cup_{i \in [t]} \cV_i$, all $\cV_i$ are disjoint, $|\cA \cap \cV_i| = 1$ for all $i \in [t]$, and $\cE_i$ is a set of edges induced by $\cV_i$ in $\graph$.
For a partition $\mathfrak{P}$, let $\cH_i(\mathfrak{P})$ be a side information graph of the index coding problem induced by the vertex set
$\cV_i$ in $\graph$.  
By minimizing over all such $\mathfrak{P} \in \pp$, and by applying Theorem~\ref{thm:minrank_bounds}, we obtain the following result. 

\begin{corollary}
The optimum number of transmissions for the given bipartite data dissemination problem is less or equal to
\begin{multline*}
\min_{\mathfrak{P} \in \pp} \sum_{i=1}^t \MINRANK_2 \left( \cH_i(\mathfrak{P}) \right) \le 
\min_{\mathfrak{P} \in \pp} \sum_{i=1}^t \clc \left( \cH_i(\mathfrak{P}) \right) \\ 
= \min_{\mathfrak{P} \in \pp} \clc \left( \cH(\mathfrak{P}) \right)
\; ,
\end{multline*}
where $\cH(\mathfrak{P})$ is the graph obtained by the union of the graphs $\cH_i(\mathfrak{P})$, $i \in [t]$. 
\end{corollary}

\section{Data exchange protocol extension to many rounds}
\label{section:dep_improvements}

In this section, we consider a more general scenario. Here, the underlying network graph $\graph(\cV, \cE)$ is an arbitrary directed graph. 
For each node $\ell \in \cV$, we require that $\cP_i \cup \cT_i = [n]$. 

In the Section~\ref{sec:one_round}, only a single round of communications was
required. If the underlying network graph is an arbitrary directed graph, then
several rounds of communications may be needed. The goal is to minimize the
total number of communicated bits while the number of rounds is bounded.

\begin{proposition}
For a node $\ell \in \cV$, and for $i \in [n]$, denote by $\distance_\ell(x_i)$ the length of the shortest path 
from a set of vertices having $x_i$ in their possession to $\ell$. Let  
$\distance_\ell = \sum_{i \in \cT_\ell} \distance_\ell(x_i)$ and 
\begin{equation}
    \distance_{\max} = \max_{\ell \in \cV} \distance_\ell \; . 
\label{eq:maximize}
\end{equation}
Then, the minimum number of transmissions in any protocol for data dissemination problem is at least $\distance_{\max}$. 
\label{prop:lower-bound}
\end{proposition}
The proof of this proposition appears in the appendix. 
\medskip 

Let the matrix families $\AAA_i$, for $i \in \cV$, be as defined in~(\ref{eq:A-i}), and 
$\AAA$ be as defined in~(\ref{eq:possession-matrix}). 

\begin{definition}
    The maximum rank of the matrix family $\AAA$ is defined as
    \begin{align*}
        \MAXRANK (\AAA) = \max\limits_{\bldAA \in \AAA} \rank (\bldA).
        \label{eq:family_rank}
    \end{align*}
    \label{definition:maxrank}
\end{definition}

Given the matrix family $\AAA_i$, we define an operator $\Gamma(\cdot)$, which replaces the symbols `$\star$' in the maximal number of the first rows with
linearly independent canonical vectors, and replaces the symbols `$\star$' in the remaining rows with zeros.

Similarly, operator $\Gamma_\ell(\cdot)$, $\ell \in \cV$, takes as an input the possession matrix from $\AAA$ and returns
$\Gamma_\ell(\AAA) = \Gamma(\AAA_\ell)$.

\begin{example}
    For a fixed $\ell \in \cV$, let
    \[
        \AAA_i = \left[\begin{matrix}
            \star & 0 & \star & \star & 0 \\
            \star & 0 & \star & \star & 0 \\
						\star & 0 & \star & \star & 0 \\
						\star & 0 & \star & \star & 0 \\ 
						\star & 0 & \star & \star & 0 \\        
						\end{matrix}\right] \ .
    \]

    After replacing the symbols `$\star$', we obtain
    \[
        \Gamma(\AAA_i) = \left[\begin{matrix}
                1 & 0 & 0 & 0 & 0\\
                0 & 0 & 1 & 0 & 0\\
                0 & 0 & 0 & 1 & 0\\
								0 & 0 & 0 & 0 & 0\\
								0 & 0 & 0 & 0 & 0\\
        \end{matrix}\right] \; .
    \]
\end{example}

\subsubsection*{Algebra of matrix families}

Denote by $\ffs$ the alphabet $\FFF \cup \{\mbox{`$\star$'}\}$. In what follows, we represent 
families of matrices over $\FFF$ as matrices over $\ffs$. In the sequel, we 
define operations on the matrices over $\ffs$, in a way which allows 
to describe algebraically the data dissemination in the network. In particular, we define two operations, 
the addition `$+$' and the multiplication `$\cdot$' of two elements $a, b \in \ffs$, in such 
way that if $a, b \in \FFF$, then these operations coincide with usual addition and 
multiplication in the field $\FFF$. 

Addition and multiplication of two elements, where at least one of the elements is `$\star$', are given in the following tables. 

Addition table: 
    \begin{align}
        \begin{array}{|c||c|c|}
            \hline \boldsymbol{+} & b & \star \\
            \hline \hline a & a + b & \star \\
            \hline \star & \star & \star \\
            \hline
        \end{array} \;\; .
        \label{table:star_addition}
    \end{align}
		
Multiplication table: 
\begin{align}
        \begin{array}{|c||c|c|c|}
            \hline \boldsymbol{\cdot} & 0 & b \neq 0 & \star \\
            \hline \hline 0 & 0 & 0 & 0  \\
            \hline a \neq 0 & 0 & a \cdot b & \star \\
						\hline \star & 0 & \star & \star \\
            \hline
        \end{array} \;\; ,
        \label{table:star_multiplication}
    \end{align}
 where $a$ and $b$ are any two elements in $\ff$.
    
The addition and multiplication operations over $\ffs$ can be naturally extended to operations on matrices over $\ffs$.  

\begin{example}
    Let a $3 \times 3$ matrix $\bldB$ over the field $\FFF$ and a $3 \times 3$ matrix family $\AAA$ over $\FFF$ be given by
    \[
        \bldB = \left[ \begin{matrix}
                1 & 1 & 0 \\
                1 & 1 & 1 \\
                0 & 1 & 1
        \end{matrix} \right] \qquad \mbox{and} \qquad
        \AAA = \left[ \begin{matrix}
                \star & 0 & 0 \\
                0 & 0 & \star \\
                0 & \star & 0
        \end{matrix} \right] \;\; .
    \]
		By multiplying the matrix family $\AAA$ from the left by the matrix $\bldB$, we obtain: 
    \begin{multline*}
        \bldB \cdot \AAA = \left[ \begin{matrix}
                1 & 1 & 0 \\
                1 & 1 & 1 \\
                0 & 1 & 1
        \end{matrix} \right] \cdot 
        \left[ \begin{matrix}
                \star & 0 & 0 \\
                0 & 0 & \star \\
                0 & \star & 0
        \end{matrix} \right] 
				= \left[ \begin{matrix}
                \star & 0 & \star \\
                \star & \star & \star \\
                0 & \star & \star \\
        \end{matrix} \right] \;\; . 
    \end{multline*}
    \label{example:matrix_multiplication}
\end{example}

We also define the multiplication of integer matrices by the matrix families over $\FFF$.

\begin{remark}
    In what follows, we use a binary operation of matrix multiplication, where one of the arguments is an integer matrix
			and the second argument is a family of matrices over $\ff$, and the result is a family of matrices over $\ff$. 
			In order to be able to do so, by slightly abusing the notation, we use the product of
    an integer matrix with a matrix over $\FFF^*$, according to the rules 
		defined in~(\ref{table:star_addition}) and~(\ref{table:star_multiplication}). 
		The result of this operation is a matrix over $\FFF^*$, which can be interpreted as a family of matrices over $\ff$.
    \label{note:integer_multiplication}
\end{remark}

\begin{example}
    Let a $3 \times 3$ integer matrix $\bldB$ be
    \[
        \bldB = \left[ \begin{matrix}
                1 & 2 & 0 \\
                4 & 5 & 6 \\
                0 & 7 & 8
        \end{matrix} \right] \;\; ,
    \]
    and $\AAA$ be a $3 \times 3$ matrix family over $\FFF$ as in Example~\ref{example:matrix_multiplication}.

    Multiplying $\bldB$ by $\AAA$ yields
    \begin{align*}
        \bldB \cdot \AAA = 
        \left[ \begin{matrix}
                1 & 2 & 0 \\
                4 & 5 & 6 \\
                0 & 7 & 8
        \end{matrix} \right] \cdot 
       \left[ \begin{matrix}
                \star & 0 & 0 \\
                0 & 0 & \star \\
                0 & \star & 0
        \end{matrix} \right] 
        = \left[ \begin{matrix}
                \star & 0 & \star \\
                \star & \star & \star \\
                0 & \star & \star \\
        \end{matrix} \right] \; . 
    \end{align*}
\end{example}

\subsection{The role of adjacency matrix}

\begin{lemma}
    \label{lemma:family_change}
    Let $\AAA$ be the possession matrix as defined in
    Equation~(\ref{eq:possession-matrix}). Let $\bldD$ be the adjacency
    matrix of the graph $\graph$. Let $\bldE$ be an $n \times n$ identity matrix. 
		It is possible to choose integers $\tau_{i,\ell}$ and coefficients $\mu_{x, i, \ell, j}$ and $\mu_{x, \ell, j}$ in Protocol $\sP$, such that after performing  one round of the protocol, the new
    possession matrix $\AAA_+$ is related to $\AAA$ as
    \[
        \AAA_+ = (\bldD \otimes \bldE) \cdot \AAA.
    \]
\end{lemma}
\begin{proof}
    From the definition of $\AAA$ in Equation~(\ref{eq:possession-matrix}), the matrix families
    $\AAA_i$, $i \in \cV$, have $n$ identical rows. 
		Then, we can write $\AAA_i = \AAA_i^{[1]} \otimes \bldone_n$, where $\AAA_i^{[1]}$ is a row vector over $\ffs$ of length $n$ which consists of a single row of $\AAA_i$. From the definition of the tensor product, we have $\AAA = \hat{\AAA} \otimes \bldone_n$, 
    where
    \begin{align*}
        \hat{\AAA} = \left[ \begin{array}{c}
            \AAA_1^{[1]} \\
						\hline
            \AAA_2^{[1]} \\
            \hline
						\vdots \\
            \hline
						\AAA_k^{[1]} \\
            \end{array} \right] \; .
    \end{align*}

    The right hand side of the claim can be re-written as
    \begin{align*}
        (\bldD \otimes \bldE) \cdot \AAA & = (\bldD \otimes \bldE) \cdot (\hat{\AAA} \otimes \bldone_n) \\
        & \stackrel{(*)}{=} (\bldD \cdot \hat{\AAA}) \otimes (\bldE \cdot \bldone_n) \\
        & = (\bldD \cdot \hat{\AAA}) \otimes (n \bldone_n) \\
        & \stackrel{(**)}{=} (\bldD \cdot \hat{\AAA}) \otimes \bldone_n \; .
    \end{align*}

    Here, the transition $(*)$ is due to the properties of the tensor product, and
    the transition $(**)$ is due to Remark~\ref{note:integer_multiplication}.
		
    Next, assume that 
		\begin{equation*}
		\hat{\AAA} = \Big(\hat{a}_{\ell,j} \Big)_{\substack{\ell \in \cV\\j \in [n]}} ,  \;
		\bldD=\Big(d_{\ell,j}\Big)_{\substack{\ell \in \cV\\j \in \cV}}  \mbox{ and } 
		\bldD \cdot \hat{\AAA} = \Big( \theta_{\ell,j} \Big)_{\substack{\ell \in
        \cV\\j \in [n]}} \; .
		\label{eq:matrices}
		\end{equation*}
		
    By using tables in~(\ref{table:star_addition}) and~(\ref{table:star_multiplication}), for all $i \in \cV, j \in [n]$, we have
    \begin{equation}
        \theta_{\ell,j} = \sum_{i \in \cV} d_{\ell,i} \cdot \hat{a}_{i,j}.
    \label{eq:theta-entry}
		\end{equation}

    Assume that there is an edge $(i,\ell) \in \cE$ for some $i \in \cV$, and that the node $i$ has $x_j$.
    Then, $d_{\ell,i} \neq 0$ and $\hat{a}_{i,j} = \mbox{`$\star$'}$. In that case, we obtain $\theta_{\ell,j} = \mbox{`$\star$'}$. 
		This correctly represents the situation that the node $i$ delivers $x_j$ to the node $\ell$. 

    We conclude that the matrix $(\bldD \cdot \hat{\AAA}) \otimes \bldone_n$ correctly represents the
    possession matrix of the graph $\graph$ after one round of execution of Protocol $\sP$.
\end{proof}
\medskip

Lemma~\ref{lemma:family_change} can be naturally extended to protocols with several communications rounds. 
In the sequel, we denote by $\AAA^{(i)}$, $i \in \nn$, the possession matrix after the $i$-th round of the protocol.
For the sake of convenience, we also use the notation $\AAA^{(0)} = \AAA$. 

\begin{corollary}
    It is possible to choose integers $\tau_{i,\ell}$ and coefficients $\mu_{x, i, \ell, j}$ and $\mu_{x, \ell, j}$ in Protocol $\sP$, such that the possession matrix after the $i$-th round of the protocol execution is given by
    \[
        \AAA^{(i)} = (\bldD^i \otimes \bldE) \cdot \AAA^{(0)}.
    \]
\label{cor:possession-matrix}
\end{corollary}

\begin{figure*}[t]
\begin{center}
    \begin{tabular}{l|ccccccc}
        \multicolumn{1}{c}{} & & & & & & & \\[\dimexpr-\normalbaselineskip-\arrayrulewidth]
        \hline
				\hline
        Range & $[1, 1.2)$ & $[1.2, 1.4)$ & $[1.4, 1.6)$ & $[1.6, 1.8)$ & $[1.8, 2.0)$ & $[2.0, \infty) $ \\
        \hline
        Occurrence, \% & 54 & 22 & 6 & 4 & 0 & 14 \\
				\hline
				\hline
    \end{tabular}
\caption{The efficiency of the algorithm for graphs of diameter $2$}
\label{fig:tworound}
\end{center}
\end{figure*}

\begin{figure*}[t]
 \begin{center}
   \begin{tabular}{l|ccccccc}
        \multicolumn{1}{c}{} & & & & & & & \\[\dimexpr-\normalbaselineskip-\arrayrulewidth]
        \hline
        \hline
				Range & $[1, 1.2)$ & $[1.2, 1.4)$ & $[1.4, 1.6)$ & $[1.6, 1.8)$ & $[1.8, 2.0)$ & $[2.0, \infty) $ \\
        \hline
				Occurrence, \% & 30 & 18 & 24 & 0 & 6 & 22 \\
				\hline
				\hline
    \end{tabular}
\caption{The efficiency of the algorithm for graphs of diameter $3$}
\label{fig:threeround}
\end{center}
\end{figure*}

\subsection{Data dissemination using rank optimization}

The following theorem is the main result of this section.

\begin{theorem}
    Let $\graph$ be an underlying directed graph of an $r_0$-solvable network defined by the
    adjacency matrix $\bldD^T$. Let $\AAA$ be the corresponding possession matrix of
    the network. Then there exists, for a suitable choice of $\tau_{i,\ell}$, $\mu_{x, i, \ell, j}$ and $\mu_{x, 1, \ell, j}$, Protocol $\sP$ with $r$ rounds, for any $r \ge r_0$, and $\tau$ transmissions, where
    \begin{equation}
        \tau \; = \; \sum_{i=1}^r \left( \min_{\bldAA^{(i)} \in (\bldDD^{i-1} \otimes
        \bldEE) \cdot \AAA} \left\{ \sum_{j \in \cV} \rank \left( \bldA_j^{(i)} \right) \right\} \right)
    \label{eq:num_transmissions}
		\end{equation}
    for matrices $\bldA^{(i)}$ which are subject to 
    \begin{multline} 
        \forall j \in \cV \, : \; \rank \left( \left[ \begin{array}{c} 
                \left(\diag \left( \bldD^{[j]}\right) \otimes \bldI \right) \cdot \bldA^{(i)} \\
								\hline
								\Gamma_j(( \bldD^{i-1} \otimes \bldE) \cdot \AAA) 
        \end{array} \right] \right) \\
        = 
        \MAXRANK \left( \left(\diag ( \blde_j ) \otimes \bldI \right) \cdot 
        \left( \bldD^i \otimes \bldE \right) \cdot \AAA \right) \, , 
        \label{equation:dep_condition}
    \end{multline}
    where the matrices $\bldI$ and $\bldE$ are both $n \times n$.
    \label{theorem:dep_impr}
\end{theorem}
\vspace{-1ex}
The proof of this theorem appears in the appendix. 
\medskip

We note that Theorem~\ref{theorem:dep_impr} is an \emph{existence} result, and therefore it gives an upper bound on the optimal number 
of transmissions. It does not claim any optimality, though. In general, it is possible that the number of transmissions given by~(\ref{eq:num_transmissions})
is not optimal. 

\section{Experimental results}

In this section, we describe experimental study of the tightness 
of the bound in Theorem~\ref{theorem:dep_impr}. 
The instance of the problem consists of two main ingredients: the adjacency
matrix of the graph and the possession matrix of the network. 
We generate the adjacency matrix of the graph randomly, while fixing the number of vertices in
the graph and the diameter. We also generate randomly the possession matrix of the network. 

In general, enumeration of all the matrices in a matrix family has exponential complexity.
In order to facilitate this process, we use a randomized algorithm. It picks random matrices 
from a given matrix family, and then checks if that matrix satisfies the conditions of the theorem. 
We use two different types of networks: in the first case the diameter of the graph $\graph$ was two, 
and in the second case it was three. In both cases, the number of nodes was $4$ and the number of
information bits was $4$.

For each randomly chosen network, we compute the number of transmissions
guaranteed by Theorem~\ref{theorem:dep_impr} over the binary field
and the lower bound on the number of transmissions in Proposition~\ref{prop:lower-bound}. We compute the ratio of these 
two quantities. The tables in Figures \ref{fig:tworound} and \ref{fig:threeround} present the distribution of this ratio. 
In order to compute the maximum rank of a matrix family, we use the algorithm in~\cite{Geelen} (see also~\cite{Hartfiel}).

\section{Acknowledgements}

The work of the authors is supported by the Norwegian-Estonian Research Cooperation Programme through the research grant EMP133, by the Estonian Ministry of Education and Research through the research grants PUT405 and {IUT2-1}, by the European Regional Development Fund through the Estonian Center of Excellence in Computer Science, EXCS, and by the EU COST Action IC1104 on Random Network Coding and Designs over  $\ff_q$. The first author has received a scholarship through the IT Academy scholarship programme.

% \clearpage

\section{Appendix}

\begin{proof} {\it (Lemma~\ref{lemma:shortest-path})}
\begin{enumerate}
\item
Consider a protocol, where in each round, each node broadcasts all the symbols $x_j$ that it has in its possession (including the messages that it received in the previous rounds). Pick some $\ell \in \cV$. Let $(v_0, v_1, v_2, \cdots, v_t = \ell)$ be a shortest path from $v_0$ to $\ell$ of length $t \le r$. Then, after $i$ rounds the node $v_i$ obtains all the symbols $x_j$ that $v_0$ has in its possession. Therefore,  
after $r$ rounds, $\ell$ has all the messages possessed by all the nodes in $\cV$. 
\item
Next, we show that $r-1$ rounds are not sufficient. Let $v_0$ and $\ell$ be two vertices, such that the shortest path between them 
is of length $r$. Denote this path  $(v_0, v_1, v_2, \cdots, v_r = \ell)$. Then, the shortest path between $v_0$ and $v_i$ is $i$ for all $i \in [r]$. 
Assume that $\cP_{v_0} = [n]$ and for any $i \in \cV \backslash \{ v_0 \}$, $\cP_i = \varnothing$. Assume that $n \ge 1$ and $\cT_\ell = \{ 1 \}$. 
Then, clearly, after one iteration $v_2$ can not know $x_1$ (because the shortest path from $v_0$ to $v_2$ is of length two, and the same symbol is not retransmitted within the same round.) More generally, for the same reason, for any $i \in [r-1]$, after $i$ iterations $v_{i+1}$ can not know $x_1$. 
\end{enumerate}
\end{proof}

%-----------------------------------------------------

Before we prove Theorem~\ref{thm:coding}, we formulate and prove the following two lemmas.  

\begin{lemma}{(\cite{Yossef-journal})}
\label{lemma:exists_orthogonal}
    Let $V$ be an ambient space and $W \subseteq V$ be a linear subspace.
    If there is a vector $\bldx \in V$ such that $\bldx \not\in W$, then there exists $\bldy \in W^\perp$ such that $\bldx \cdot \bldy
    \neq \bldzero$.
\end{lemma}
\begin{proof}
    Let $\bldx \in V$ be such that $\bldx \not\in W$. By contrary, assume that for all $\bldy \in W^\perp$ we have $\bldx \cdot \bldy
    = \bldzero$. 
		Then, $\bldx \in (W^\perp)^\perp = W$. This is in contradiction to the conditions of the lemma.  		
\end{proof}

\begin{lemma}
\label{lemma:orthogonal_subspace}
    Let $V$ be an ambient space and $W \subseteq V$ be its linear subspace. If $\bldx \in W^\perp$, then for every subspace $W' \subseteq W$ and for every vector $\bldy \in W'$ it holds $\bldx \cdot \bldy = \bldzero$.
\end{lemma}
\begin{proof}
    Let $\bldx \in W^\perp$. Pick any $\bldy \in W'$. Then $\bldy \in W$. We obtain that $\bldx \cdot \bldy = \bldzero$. 
\end{proof}
\medskip

\begin{proof} {\it (Theorem~\ref{thm:coding})}

    The statement of the theorem is proven in two steps. 
		
		\begin{enumerate} 
		\item
		We construct an exact coding scheme, which uses $\tau$ symbol transmissions over $\ff$. 
		We show that under this scheme, for all $\ell \in \cV$, the node $\ell$ can recover the bit $x_\eta$ for all $\eta \in \cT_\ell$.

    Let $\bldA \in \AAA$ be the matrix, which minimizes the value of $\tau$ in~(\ref{eq:tau-expression}). 
    Assume that Equation~(\ref{ex_condition}) holds, and take some $\eta \in \cT_\ell$, $\ell \in \cV$. 
		Then, 
		\[
		    \blde_\eta^T = \sum_{i \in \cN_{in} (\ell), \; j \in [n] } \alpha_{i,j} \bldA_i^{[j]} + 
				\sum_{ j \in [n] } \beta_{j} \bldP_\ell^{[j]} \; , 
    \]		
		where all $\alpha_{i,j}$ and $\beta_j$ are in $\ff$. Then, 
\begin{multline*}
		    x_\eta \; = \;  \blde_\eta^T \bldx  \\
				= \; \sum_{i \in \cN_{in} (\ell), \; j \in [n] } \alpha_{i,j} (\bldA_i^{[j]} \cdot \bldx) + 
				\sum_{ j \in [n] } \beta_{j} (\bldP_\ell^{[j]} \cdot \bldx) \; . 
\end{multline*}		
     Each (sending) node $i \in \cV$ will use some basis $\cB_i$ of the rowspace of $\bldA_i$, and will transmit the messages $(\bldb \cdot \bldx) \in \ff$ for all $\bldb \in \cB_i$. It is straightforward to verify that in such way each node transmits combinations of bits that it has in its possession. The total number of messages that the node $i$ transmits is $\rank(\bldA_i)$ and the total number of messages transmitted in the scheme is 
		\[
		\sum_{i \in \cV} \rank (\bldA_i) \; . 
		\]
		
		Each (receiving) node $\ell \in \cV$ will be able to compute the values  $(\bldA_i^{[j]} \cdot \bldx)$ for all $i \in \cN_{in} (\ell)$, $j \in [n]$, from the messages $(\bldb \cdot \bldx)$. It will also be able to compute the values $\bldP_\ell^{[j]} \cdot \bldx$ for all $j \in [n]$. Therefore, the node $\ell$ will be able to compute $x_\eta$, as required.  
		
	\item	
    We show that if there exists another linear code which satisfies the 
		requests of all the nodes in $\graph$, then it is possible to construct a corresponding matrix $\bldA$
    as in Equation~(\ref{ex_condition}), which satisfies the conditions of the theorem.

 %   Assume that a linear code is defined by the $n \times n$ transmission matrices 
%		$\bldS_\ell$, $\ell \in \cV$, and the messages transmitted by the node $\ell$ 
%		are of the form $\bldS_\ell^{[i]} \cdot \bldx$, $i \in [n]$. 
		
		Consider the transmission scheme with the optimal number of transmissions $\tauopt$. Assume that 
		for each $\ell \in \cV$, the node $\ell$ transmits $n_\ell$ messages of the form 
		$\blds_\ell^{(i)} \cdot \bldx$, $i \in [n_\ell]$, $0 \le n_\ell \le n$.  
		Here $\sum_{\ell \in \cV} n_\ell = \tauopt$ is the total number of transmissions. 
		
%		The number of messages 
%		transmitted by the node $\ell$ can be smaller than $n$. In that case, we 
%		append a number of zero rows to $\bldS_\ell$, such that the total size is $n \times n$. 
		
%    As the transmitter $T_i$, $i \in [k_T]$ has only bits $\pi_T(i)$
%    of the information vector, then the vectors in $S_i$ can have non-zero
%    coefficients at coordinates at locations $\pi_T(i)$.

 Next, we show that for all $\ell \in \cV$ and for all $\eta \in [n]$, if $\eta \in \cT_\ell$ 
then the vector $\blde_\eta \in \ff^n$ belongs to $W_\ell \subseteq \ff^n$, where $W_\ell$ is the linear span of the vector set 
\[
\left( \bigcup_{\stackrel{j \in \cN_{in}(\ell)}{i \in [n_j]}} \left\{ \blds_j^{(i)} \right\} \right) \; \cup \; \left( \bigcup_{ j \in \cP_\ell} \left\{ \blde_j \right\} \right) \; . 
\]  
Fix some $\ell \in \cV$. By contrary, assume that $\blde_\eta \not\in W_\ell$. Then, by Lemma~\ref{lemma:exists_orthogonal}
there exists $\bldx \in W_\ell^\perp$ such that $\blde_\eta \cdot \bldx \neq 0$.

%   If the node $\ell \in \cV$ requests the bit $k \in [n]$, then $k \in \cT_\ell$, and
%    thus we can construct a vector $Q_j^{(m)}$ with $1$ in the $m$-th coordinate
%    and zeros elsewhere. By contradiction, if for some $j \in [k_R]$ and $m \in
%    [n]$ there is $Q_j^{(m)} \not\in \vspan S_1 \oplus \cdots \oplus \vspan
%    S_{[k_T]} \oplus \rowspace P_j =: W$, then by
%    Lemma~\ref{lemma:exists_orthogonal} there exists $\xx \in W^\TRANS$ such
%    that $Q_j^{(m)} \cdot \xx \neq 0$.

From the definition of $W_\ell$ and Lemma~\ref{lemma:orthogonal_subspace}, we have
    that $\bldx \cdot \blds_j^{(i)} = 0$ for all $j \in \cN_{in}(\ell)$, $i \in [n_j]$, and $\bldx \cdot \blde_j = 0$ for all
		$j \in \cP_\ell$. This means that:
    \begin{enumerate}
        \item[(i)] the transmitted messages $\bldx \cdot \blds_j^{(i)}$ are $0$ for every transmitter $j \in \cN_{in}(\ell)$, $i
            \in [n_j]$;
        \item[(ii)] the side information symbols $x_{i}$, $i \in \cP_\ell$, available to the node $\ell$ are all $0$.
    \end{enumerate}
    
    Thus, the node $\ell$ cannot distinguish between the information vector $\bldx$ and the zero vector $\bldzero$. However, $x_\eta \neq 0$. 
		Therefore, our assumption that $\blde_\eta \not\in W$ is false. We conclude that $\blde_\eta \in W_\ell$.

    Next, we construct the $n \times n$ matrices $\bldA_\ell$ for all $\ell \in \cV$. For that sake, we take 
		\[
		\bldA_\ell^{(i)} = \left\{ \begin{array}{ll} 
		  \blds_\ell^{(i)} & \mbox{ if } i \in [n_\ell] \\
		  \bldzero & \mbox{ otherwise }
			\end{array} \right. 
		\]
		We obtain that for each $\ell \in \cV$, $\rank(\bldA_\ell) \le n_\ell$, and therefore  
		\[
		\sum_{\ell \in \cV} \rank \left( \bldA_\ell \right) \le \tauopt \; . 
		\]
		
    By construction, the resulting $\bldA$ belongs to the family $\AAA$, and therefore the corresponding code satisfies 
    equation~(\ref{ex_condition}). We conclude that $\tau$ in expression~(\ref{eq:tau-expression}) is indeed the minimum number of
    transmissions.
    \end{enumerate}
\end{proof}
\medskip

%-----------------------------------------------------
\begin{proof} {\it (Proposition~\ref{prop:lower-bound})}

The proof of this proposition is straightforward: let $\ell \in \cV$ be the node that maximizes the expression~(\ref{eq:maximize}). 
Then, at least $\distance_\ell$ transmissions are needed in order to satisfy all the requests of $\ell$. 
\end{proof} 
\medskip
%-----------------------------------------------------

\begin{proof} {\it (Corollary~\ref{cor:possession-matrix})}

    We have:
    \begin{align*}
        \AAA^{(i)} & = (\bldD \otimes \bldE)^i \cdot \AAA^{(0)} \\
        & = (\bldD^i \otimes \bldE^i) \cdot \AAA^{(0)} \\
        & = (\bldD^i \otimes n^{i-1} \bldE) \cdot \AAA^{(0)} \\
        & = n^{i-1} (\bldD^i \otimes \bldE) \cdot \AAA^{(0)} \\
        & \stackrel{(\S)}{=} (\bldD^i \otimes \bldE) \cdot \AAA^{(0)}
    \end{align*}
    Here, the transition $(\S)$ holds due to Remark~\ref{note:integer_multiplication}. Thus, 
    any non-zero integer entry in $(\bldD^i \otimes \bldE)$ is mapped to the element $1 \in \FFF$, 
		and, therefore, the factor $n^{i-1}$ can be omitted.
\end{proof}
\medskip

Before we turn to proving Theorem~\ref{theorem:dep_impr}, we formulate and prove the following lemma.

\begin{lemma}
    \label{lemma:existance}
    Let $\graph$ be a directed graph defined by the adjacency matrix $\bldD^T$. Let
    the possession matrix family of the graph $\graph$ be $\AAA$ as defined in
    Equation~(\ref{eq:possession-matrix}). There exists a transmission
    matrix $\bldA \in \AAA$ such that 
    \begin{multline}
    \rank \left( \left[ \begin{array}{c} 
                \left(\diag \left( \bldD^{[j]}\right) \otimes \bldI \right) \cdot \bldA \\
								\hline
								\Gamma_j(\AAA) 
        \end{array} \right] \right) \\
        = 
    \MAXRANK \left( \left(\diag ( \blde_j ) \otimes \bldI \right) \cdot 
        \left( \bldD \otimes \bldE \right) \cdot \AAA \right) 
    \label{eq:exists-A}
		\end{multline}
    for all $j \in \cV$.
\end{lemma}
\begin{proof}
We analyze the left and the right-hand side of equation~(\ref{equation:dep_condition}) separately. 
\begin{enumerate}
\item
    The right-hand side of equation~(\ref{equation:dep_condition}) can be written as
    \begin{eqnarray*}
        && \hspace{-8ex} (\diag(\blde_j) \otimes \bldI) \cdot (\bldD \otimes \bldE) \cdot \AAA \\
				& = & (\diag(\blde_j) \otimes \bldI) \cdot (\bldD
        \otimes \bldE) \cdot (\hat{\AAA} \otimes \bldone_n) \\
        & = & (\diag(\blde_j) \cdot \bldD \cdot \hat{\AAA}) \otimes (\bldI \cdot \bldE \cdot \bldone_n) \\
        & = & (\diag(\blde_j) \cdot \bldD \cdot \hat{\AAA}) \otimes (n \bldone_n) \\
        & \stackrel{(\P)}{=} & (\diag(\blde_j) \cdot \bldD \cdot \hat{\AAA}) \otimes \bldone_n \; . 
    \end{eqnarray*}
    The equality $(\P)$ holds because $n>0$ and all non-zero integers are mapped
    to field element $1$, thus we can omit the factor $n$.

    By employing the notation in~(\ref{eq:matrices}), the equation~(\ref{eq:theta-entry}) holds. Then, 
		the $j$-th row of the matrix $\diag(\blde_j) \cdot \bldD \cdot \hat{\AAA}$ is $(\theta_{j,\eta})_{\eta \in [n]}$. We have:
    \begin{multline}
        \diag(\blde_j) \cdot \bldD \cdot \hat{\AAA} = \\
				\left[ \begin{array}{ccc}
            0 &  \ldots & 0 \\
            \vdots & \ddots & \vdots \\
            0 &  \ldots & 0 \\
            \sum_{i \in \cV} d_{j,i} \cdot \hat{a}_{i,1} &  \dots & \sum_{i \in \cV} d_{j,i} \cdot \hat{a}_{i,n} \\
            0 & \ldots & 0 \\
            \vdots &  \ddots & \vdots \\
            0 &  \ldots & 0
        \end{array} \right] \; .
        \label{equation:theorem_rhs}
    \end{multline}

    Next, it is straightforward to verify that the $\MAXRANK$ of the matrix family 
		$(\diag(\blde_j) \cdot \bldD \cdot \hat{\AAA}) \otimes \bldone_n$ is the
    number of the symbols `$\star$' in the non-zero row of the matrix in
    Equation~(\ref{equation:theorem_rhs}). 

\item
    Consider the upper-block part of the matrix in the left-hand side of~(\ref{equation:dep_condition}). Denote 
		\[
		\bldA = \left( \tilde{a}_{i,j} \right)_{\stackrel{i \in [kn]}{j\in [n]}} \; . 
		\]
		In the sequel, we show that the values of the elements $\tilde{a}_{i,j}$ in $\bldA$ can be chosen such that the equation~(\ref{equation:dep_condition}) holds.  

    The matrix $\diag(\bldD^{[j]}) \otimes \bldI$ is the diagonal block matrix, namely, 
    \[
        \diag(\bldD^{[j]}) \otimes \bldI =
        \left[ \begin{array}{c|c|c|c}
            \bldD_{j,1} & \bldzero_n & \cdots & \bldzero_n \\
            \hline
						\bldzero_n & \bldD_{j,2} &\cdots & \bldzero_n \\
  					\hline
            \vdots & \vdots & \ddots & \vdots \\
            \hline
            \bldzero_n & \bldzero_n & \cdots & \bldD_{j,k}  
        \end{array}\right] \; , 
    \]
where for all $i \in \cV$, $\bldD_{j,i}$ is a diagonal $n \times n$ matrix as follows: 
    \[
    \bldD_{j,i} = \left[ \begin{array}{cccc}
            d_{j,i} & 0 & \cdots & 0 \\
            0 & d_{j,i} & \cdots & 0 \\
            \vdots & \vdots & \ddots & \vdots \\
            0 & 0 & \cdots & d_{j,i} \\
        \end{array}\right] \; , 
		\]
		and $\bldzero_{n}$ is an $n \times n$ all-zero matrix. 

    Then, $\; (\diag(\bldD^{[j]}) \otimes \bldI) \cdot \bldA = $
    \begin{equation}
        \hspace{-0.8ex} \left[ \begin{array}{ccc}
            d_{j,1} \cdot \tilde{a}_{1,1} & \cdots & d_{j,1} \cdot \tilde{a}_{1,n} \\
            \vdots & \ddots & \vdots \\
            d_{j,1} \cdot \tilde{a}_{n,1} & \cdots & d_{j,1} \cdot  \tilde{a}_{n,n} \\
						&& \\
						\hline
            && \\
						\vdots & \ddots & \vdots \\
            && \\
						\hline
						&& \\
						d_{j,k} \cdot \tilde{a}_{(k-1)n+1,1} & \cdots & d_{j,k} \cdot \tilde{a}_{(k-1)n+1,n} \\
            \vdots & \ddots & \vdots \\
            d_{j,k} \cdot \tilde{a}_{kn,1} & \cdots & d_{j,k} \cdot \tilde{a}_{kn,n} \\
        \end{array} \right] .
        \label{equation:theorem_lhs}
    \end{equation}

    The element in the $j$-th row and $\ell$-th column in the matrix in
    Equation~(\ref{equation:theorem_rhs}) is `$\star$' if there exists $i$ such
    that $d_{j,i} \neq 0$ and $\hat{a}_{i,\ell}= \mbox{`$\star$'}$. In that case, we can pick
    $s \in [n]$ and set $\tilde{a}_{(i - 1)n + s, \ell}$ to $1$. We obtain that
    $d_{j, i} \cdot \tilde{a}_{(i-1)n + s, \ell} \neq 0$.
    
    Since $s \in [n]$, different $s$ can be chosen for every $\ell \in [n]$. 
		After setting $\tilde{a}_{(i - 1)n + s, \ell}$ to $1$ in every column $\ell \in [n]$, 
		we set the values of all other elements in $\bldA$ to $0$. 
		Because the ones in the matrix in Equation~(\ref{equation:theorem_lhs}) are all in the distinct rows and in
    the distinct columns, the rank of the matrix $(\diag(\bldD^{[j]}) \otimes \bldI) \cdot \bldA$
    equals to the $\MAXRANK$ of the family $(\diag(\blde_j) \cdot \bldD \cdot \hat{\AAA}) \otimes \bldone_n$.

    If $\hat{a}_{j, \ell}=\mbox{`$\star$'}$, then $a_{(j-1) n + s, \ell}= \mbox{`$\star$'}$, for $s \in
    [n]$, and only these elements are set to $1$ in $\bldA$. Therefore, $\bldA \in \AAA$.

  From the construction of $\bldA \in \AAA$, we have that $\left(\diag
    \left( \bldD^{[j]} \right) \otimes \bldI \right) \cdot \bldA$ has a single
    one in some row, for every column where there is \mbox{`$\star$'} in
    $\left(\diag ( \blde_j ) \otimes \bldI \right) \cdot \left( \bldD \otimes
    \bldE \right) \cdot \AAA$. The transposed adjacency matrix $\bldD$ has ones in the
    main diagonal. Therefore, if there exists a column with `$\star$' in
    $(\diag(\blde_j) \cdot \bldD \cdot \hat{\AAA}) \otimes \bldone_n$, then there
    also exists \mbox{`$\star$'} in the same column of $\AAA$.
    
    Thus, if there is a single one in a row in any column of $\left(\diag
    \left( \bldD^{[j]}\right) \otimes \bldI \right) \cdot \bldA$, then there
    is a single one in a row in the corresponding column of $\Gamma_j(\AAA)$. As
    \begin{multline*}
     \rank\left(\left(\diag \left( \bldD^{[j]}\right) \otimes
     \bldI \right) \cdot \bldA\right) \\ 
		= \; \MAXRANK \left( \left(\diag ( \blde_j ) \otimes
     \bldI \right) \cdot \left( \bldD \otimes \bldE \right) \cdot \AAA \right) \; ,
    \end{multline*}
    then 
		\begin{multline*}
		\rowspace\left(\Gamma_j(\AAA)\right) \\
		\subseteq \; \rowspace \left( \left(\diag \left(
    \bldD^{[j]}\right) \otimes \bldI \right) \cdot \bldA\right) \; , 
		\end{multline*}
		and condition~(\ref{eq:exists-A}) holds.
\end{enumerate}
\end{proof}
\medskip

The proof of the last lemma showed that the transmission matrix $\bldA$ exists. 
However, it may not be optimal. We next turn to proving Theorem~\ref{theorem:dep_impr}.
\medskip

\begin{proof}{\it (Theorem~\ref{theorem:dep_impr})} 

    From Lemma~\ref{lemma:existance}, there exist matrices $\bldA^{(i)}$ satisfying~(\ref{equation:dep_condition}).
    Take any such matrices, and write them as 
		\[
		\bldA^{(i)} = \Big( a^{(i)}_{\rho,\eta} \Big)_{\mbox{\scriptsize $\begin{array}{c} \rho \in [kn] \\ \eta \in [n] \end{array}$}} \; . 
		\]

    For all $s \in \cV$, let the vectors 
    \[
		\boldsymbol{t}^{(i)}_{s,r}=\left( t^{(i)}_{s,r,1}, t^{(i)}_{s,r,2}, \ldots,t^{(i)}_{s,r,n} \right) \; 
		\] 
		be the $r$-th row of $\bldA^{(i)}_s$. These vectors can be viewed as the linear
    coefficients multiplying the symbols transmitted by the node $s$ during the $i$-th
    round of the protocol. In the $i$-th round of the protocol, the
    messages transmitted by the node $s$ are given by the non-zero vectors in 
    \begin{equation}
        \bldUpsilon^{(i)}_{s} = \left\{ \sum_{m \in [n]} t^{(i)}_{s,r,m} \cdot x_m \right\}_{r \in [n]} 
        = \Bigg\{ \boldsymbol{t}^{(i)}_{s,r} \cdot \boldsymbol{x}  \Bigg\}_{r \in [n]} \; ,
        \label{equation:theorem_transmitted_messages}
    \end{equation}
    where $\boldsymbol{x} = (x_1, x_2, \cdots, x_n)^T$. 

    The number of transmissions of the node $s$ during the $i$-th round of the protocol is the rank of
    $\bldA^{(i)}_s$. When summing for all $s \in \cV$, we obtain the number of transmissions $\tau$ as stated in
    the right-hand side of Theorem~\ref{theorem:dep_impr} with respect to this $\bldA^{(i)}$. 

		Observe that the node $\ell$ receives all the messages 
		from the node $s$ if $d_{\ell,s} \neq 0$. 
    Therefore, the node $\ell$ receives all the messages of the form 
    \[
        d_{\ell,s} \cdot \sum_{m \in [n]} \left( t^{(i)}_{s,r,m} \cdot x_m \right) \; ,
    \]
    for all $s \in \cV$, $r \in [n]$.

    Since $t^{(i)}_{s,r,m} = a^{(i)}_{(s-1)n+r,m}$, the messages
    received by the node $t$ are the entries of the vector given by: 
    \begin{align}
        \left[ \begin{array}{ccc}
            d_{1,1} \cdot a^{(i)}_{1,1} & \cdots & d_{1,1} \cdot a^{(i)}_{1,n} \\
            \vdots & \dots & \vdots \\
            d_{1,1} \cdot a^{(i)}_{n,1} & \cdots & d_{1,1} \cdot a^{(i)}_{n,n} \\
            && \\
						\hline
						&& \\
						\vdots & \ddots & \vdots \\
            && \\
						\hline
						&& \\
						d_{k,1} \cdot a^{(i)}_{(k-1)n+1,1} & \cdots & d_{k,1} \cdot a^{(i)}_{(k-1)n+1,n} \\
            \vdots & \ddots & \vdots \\
            d_{k,1} \cdot a^{(i)}_{kn,1} & \cdots & d_{k,1} \cdot a^{(i)}_{kn,n} 
        \end{array} \right] \cdot 
        \left[ \begin{array}{c}
            x_1 \\
						x_2 \\
            \vdots \\
            x_n
        \end{array} \right] \; .
        \label{equation:received_matrix}
    \end{align}

    From Lemma~\ref{lemma:family_change}, the matrix family $(\bldD^i \otimes \bldE) \cdot \AAA$
    is the possession matrix of the network after the round $i$. Thus, the matrix family
    $(\diag(\blde_j) \otimes \bldI) \cdot (\bldD^i \otimes \bldE) \cdot \AAA$ is the possession
    matrix of the node $j$ after the round $i$. The $\MAXRANK$ of this matrix
    family is the number of symbols the node $j$ has after the completion of the 
    $i$-th round of the protocol.

    To this end, the matrices $\bldA^{(i)}$ as above satisfy the condition~(\ref{equation:dep_condition}), and the 
		number of transmission in the protocol based on it is given by the right-hand side of the equality~(\ref{eq:num_transmissions}). 
    Therefore, in order to minimize the number of transmissions in the protocol, one has to choose the matrices $\bldA^{(i)}$ that 
		satisfy~(\ref{equation:dep_condition}) and minimize the right-hand side of the equality~(\ref{eq:num_transmissions}).
\end{proof}

\end{document}